\documentclass[letterpaper,11pt]{article}
\usepackage{amsthm}  
\usepackage{amsmath, amssymb}
\usepackage{color}
\usepackage{fullpage}
\usepackage[numbers]{natbib}
\usepackage[noend]{algorithmic}
\usepackage{thmtools, thm-restate}
\usepackage{tikz}
\usepackage{enumerate}
\usepackage{thm-restate}
\usepackage{graphicx}
\usepackage{amsopn}
\usepackage{caption}
\usepackage{hyperref}
\usepackage{subcaption}
\usepackage{zerosum,csquotes}
\usepackage{enumitem,linegoal}
\usepackage[linesnumbered,ruled,vlined]{algorithm2e}
\usepackage{comment}
\usepackage{todonotes}
\usepackage{ctable}					

\usepackage{mathtools}
\usepackage[flushleft]{threeparttable}
\usepackage[normalem]{ulem}

\usepackage{footnote}
\makesavenoteenv{tabular}
\makesavenoteenv{table}

\usepackage[margin=1in]{geometry}

 \newtheorem{theorem}{Theorem}[section]
 \newtheorem{lemma}[theorem]{Lemma}
  \newtheorem{claim}[theorem]{Claim}

 \newtheorem{corollary}[theorem]{Corollary}
\theoremstyle{definition}
 \newtheorem{definition}[theorem]{Definition}

\makeatletter
\newif\ifqed
\def\GrabProofArgument[#1]{ #1: \egroup\ignorespaces}
\def\proof{\noindent\textbf\bgroup Proof%
	\@ifnextchar[{\GrabProofArgument}{. \egroup\ignorespaces}\global\qedtrue}

\def\qedhere{\ifmmode\tag*{\qedsign}\else\hspace*{\fill}\qedsign\medskip\fi\global\qedfalse}
\def\qedsign{$\Box$}
\makeatother

\newcommand{\OPT}{\ensuremath{\mathsf{OPT}}\xspace}
\newcommand{\PC}{\textsf{P}\xspace}
\newcommand{\NP}{\textsf{NP}\xspace}

\newcommand{\nph}{\textsf{NP-hard}\xspace}
\newcommand{\apxh}{\textsf{APX-hard}\xspace}
\newcommand{\ptas}{\textsf{PTAS}\xspace}
\newcommand{\psolve}{\textsf{P}\xspace}
\newcommand{\npsolve}{\textsf{NP}\xspace}

\newcommand{\opt}{\OPT}

\DeclareMathOperator{\dns}{dens}
\newcommand{\dens}[1]{\dns_{#1}}

\definecolor{mygreen}{RGB}{20,140,80}
\definecolor{mylightgray}{RGB}{230,230,230}

\definecolor{mygreen}{RGB}{20,140,80}
\definecolor{mydarkgray}{gray}{0.15} 
\hypersetup{
     colorlinks=true,
     citecolor= mygreen,
     linkcolor= black
}

\newcommand{\citeboth}[1]{\hypersetup{citecolor=mydarkgray}\citeauthor{#1}\hypersetup{citecolor=mygreen} \cite{#1}}

\newcounter{proccnt}

\newcommand{\konote}[1]{}

\title{\Large Polynomial-time Approximation Scheme for Minimum $k$-cut\\in Planar and Minor-free Graphs}
\author{MohammadHossein Bateni\thanks{Google Research. Email: \texttt{bateni@google.com}} \\
\and
Alireza Farhadi\thanks{University of Maryland. Email: \texttt{\{farhadi, hajiagha\}@cs.umd.edu}}
	\thanks{Supported in part by NSF CAREER award CCF-1053605, NSF AF:Medium grant CCF-1161365, NSF BIGDATA grant IIS-1546108, NSF SPX grant CCF-1822738, and two small UMD AI in Business and Society Seed Grant and UMD Year of Data Science Program Grant} \\
\and
MohammadTaghi Hajiaghayi\footnotemark[2] \footnotemark[3] 
	}

\begin{document}
	\newcommand{\ignore}[1]{}
\renewcommand{\theenumi}{(\roman{enumi})}
\renewcommand{\labelenumi}{\theenumi.}
\sloppy

%
%

\date{}

\maketitle


\begin{abstract}
The $k$-cut problem asks, given a connected graph $G$ and a positive integer
$k$, to find a minimum-weight set of edges whose removal splits $G$
into $k$ connected components.
We give the first polynomial-time algorithm with approximation factor
$2-\epsilon$ (with constant $\epsilon > 0$) for the $k$-cut problem in
planar and minor-free graphs.  Applying more complex techniques, we further improve our method and
give a polynomial-time approximation scheme for the $k$-cut problem in both planar and minor-free graphs.
Despite persistent effort, to the best of our knowledge, this is the first improvement 
for the $k$-cut problem over standard approximation factor of $2$ in any major class of graphs.

\end{abstract}

\section{Introduction}

In the {\em $k$-cut problem}, given an undirected connected graph with edge
weights, the goal is to find a minimum-weight set of edges whose removal
splits the graph into $k$ connected components. The problem is
also called the {\em $k$-way cut problem} or the {\em
  multi-component cut}. This problem is a natural generalization of
the minimum cut problem in which we want to find a minimum-weight
set of edges whose removal splits the graph into two components.

\citeboth{goldschmidt1994polynomial} proved that the $k$-cut problem is \nph when $k$ is part
of the input. In the same work, they provided an $O(n^{k^2})$ algorithm for the $k$-cut problem
which is polynomial for every fixed $k$. Better algorithms have been proposed in a series of works \cite{kamidoi2006deterministic, xiao2008improved, karger1996new, thorup2008minimum}. As of today, the best algorithm for the minimum $k$-cut problem is by \citeboth{thorup2008minimum}, and has the running time of $O(n^{2k} \log n)$. Despite these improvements, this problem is proven to be W[1]-hard when $k$ is taken as a parameter \cite{downey2003cutting}. This hardness shows that no FPT algorithm parameterized by $k$ can solve this problem  unless \PC $=$ \NP.

In terms of approximation algorithms, several approximation algorithms are known for this problem \cite{saran1995finding, Naor:2001:TPA:365411.365415, ravi2002approximating, zhao2001approximating,xiao2011tight}, however the approximation ratio of none of them is better than $2-o(1)$. In fact, a recent result by \citeboth{manurangsi2017inapproximability} shows that this problem is \nph to approximate to within $2-\epsilon$ factor assuming Small Set Expansion Hypothesis. Also, to the best of our knowledge, prior to this work, 
there was no approximation algorithm with a ratio better than 2 for any major class of graphs.  It is also worth mentioning that a recent work by \citeboth{Gupta:2018:FAB:3174304.3175483} showed that using an FPT algorithm the approximation factor of $2$ can be beaten in the $k$-cut problem. They showed that there exists a $2-\epsilon$ approximation algorithm that runs in time $2^{O(k^6)} \tilde{O}(n^4)$. However, their algorithm is not polynomial when $k$ is part
of the input.

In this paper, we first show that surprisingly the
approximation guarantee of a natural greedy algorithm is $2-\epsilon$
in planar graphs as well as graphs excluding a fixed minor, for some positive constant $\epsilon$. Later, we show
how our method can be extended to
derive a \ptas for the $k$-cut problem in minor-free and planar graphs.
This is the first result that beats the approximation factor of $2$ in polynomial-time for a major class of graphs.

\begin{theorem}\label{thm:minor-free}
There exists a natural greedy algorithm (Algorithm~\ref{alg:greedy}) with the approximation ratio of $2-\epsilon$ in minor-free graphs. Also, the same algorithm achieves the approximation ratio of $1.9968...$ in planar graphs.
\end{theorem}

The above theorem is proved in Section~\ref{sec:minor-free}.  Then we
move on to the result with the better guarantee.  The following is proved in Section~\ref{sec:ptas}.

\begin{theorem}
\label{thm:mainmain} There is a polynomial-time approximation scheme
(\ptas) for the $k$-cut problem in planar and minor-free graphs.
\end{theorem}


In our algorithm we often find minimum-weight splits (a split is subset of edges whose removal increases the number of connected components. Specifically, a $k$-way split is a split whose removal increases the number of connected components by $k-1$). Also, we work with notions of \emph{separation degree} and \emph{density of splits}. Roughly speaking, the \emph{separation degree} of a split is the number of the components that removal of a split adds to the graph. Additionally, the \emph{density} of a split is the total weight of the edges in a split divided by its separation degree. 
 
Although the density of a minimum split could be twice the density of the optimal solution, 
we show that the density of splits with a larger separation degree gets very close to the density of the optimal solution in planar and minor-free graphs.
Interestingly, the same does not hold in general graphs, 
where the density of arbitrary large splits may be as much as a factor $2-o(1)$ of the density of the optimal solution. For example, in an unweighted complete graph, the density of arbitrary large splits are $2-o(1)$ of the density of the optimal solution even if size of the splits are very large.

We show that a natural greedy algorithm that repeatedly picks a minimum-density split with a constant separation degree achieves an approximation ratio better than $2$. 
First, in order to introduce and highlight our main ideas, 
we consider the greedy algorithm which repeatedly picks a minimum-density split with a separation degree of at most $3$ and show that its approximation ratio is $2-\epsilon$ in minor-free graphs. Subsequently, we generalize our method to derive a polynomial-time approximation scheme (\ptas) in planar and minor-free graphs. 
\citeboth{saran1995finding} considers a similar greedy algorithm which successively removes the edges of a minimum cut. They showed that the approximation ratio of the greedy algorithm is $2-2/k$ in general graphs. Later, \citeboth{xiao2011tight} generalized this method by repeatedly removing the edges of a minimum $h$-way split. Although they find larger splits, they showed that the approximation ratio of this algorithm is about $2- h/k$, and it does not beat the approximation factor of $2$ by any constant factor.

In our first main result, we show that the approximation ratio of the simple greedy algorithm is better than $2$ by a constant factor in minor-free graphs. Our main observation is that in any balanced weighted graph, there exists a matching such that its weight is at least a constant fraction of the total weight of the graph. 
This result can also be viewed as a generalization of the work of \citeboth{nishizeki1979lower} in unweighted graphs. Later, we introduce a more profound analysis of our method to derive a \ptas in planar and minor-free graphs.

\subsection{Related Works}

A problem closely related to the $k$-cut problem is the multiway
cut problem. Given a set of $k$ vertices
called \emph{terminals}, in the multiway
cut problem, we want to find a minimum-weight cut that 
separates the terminals from one another.  The study of its
computational complexity was inaugurated in 1983 by Dahlhaus, Johnson,
Papadimitriou, Seymour, and
Yannakakis~\cite{dahlhaus1994complexity}\footnote{The work was first
  known in an unpublished but widely circulated extended
  abstract. Their complete paper was published in 1994.}.
  They provided a simple $2$-approximation algorithm for the multiway cut problem, and proved that the problem is \apxh for any fixed $k \geq 3$. However, in the case of planar graphs, they showed that the problem can be solved in a polynomial
  time  for fixed $k$ but is \nph when $k$ is part of the input. Surprisingly, as of today, it is not known whether $k$-cut problem is \nph in planar graphs.

The approximation factor of this problem improved in a sequence of works ~\cite{cualinescu1998improved,cunningham1999optimal,karger2004rounding}. As of today, the best approximation factor is $1.3438$ \cite{karger2004rounding}. In case of planar graphs, a very recent result by \citeboth{bateni2012polynomial} shows there exists a
\ptas for the multiway cut problem in planar graphs.

Another problem related to the $k$-cut problem is the {\em Steiner $k$-cut
problem}, which generalizes both the $k$-cut
problem and the multiway cut problem. Given an edge-weighted
undirected graph $G$, a subset of vertices $T$ called terminals, and
an integer $k \le |T|$, the objective is to find a minimum-weight set of edges
whose removal results in $k$ disconnected components, each 
containing at least one terminal. The best result known for this problem is a
$2-2/k$ approximation algorithm due to~\citeboth{chekuri2006steiner}.

We remark that the ``identity-relaxed'' variants of Steiner tree and
multiway cut problems, namely $k$-MST and $k$-cut, have been elusive
to date.  The latter problems allow us to pick the identity of $k$
``terminals'' to connect or separate, respectively.
The initial $2$-approximation algorithms for Steiner tree~\cite{GP68}
and multiway cut~\cite{dahlhaus1994complexity} were improved in a
series of work~\cite{cualinescu1998improved,KZ97,PS00,RZ05,Zelikovsky93} culminating in
a $1.3863$ approximation algorithm~\cite{BGRS13:steiner} for Steiner
tree and a $1.3438$ approximation algorithm for multiway
cut~\cite{karger2004rounding}.  Nonetheless, no approximation guarantee better
than $2 - o(1)$ is known for $k$-MST or $k$-cut. 

Similarly, in the case of planar graphs, where \ptas{}s are
known for Steiner tree~\cite{BorradaileKleinMathieu2009} and
multiway cut~\cite{bateni2012polynomial}, their identity-relaxed
variants (prior to this work) proved to be more resilient.
In particular, the standard spanner construction techniques and the
small-treewidth reduction approach developed and successfully applied
to a host of network design problems in the last decade \cite{bateni2016ptas, bateni2012polynomial, bateni2011approximation, BorradaileKleinMathieu2009, borradaile2015polynomial, borradaile2017minor, eisenstat2012efficient, klein2006subset, klein2008linear}, seem
challenging to use in this context. Recently, Cohen-Addad et
al.~\cite{CKM16} gave \ptas{}s for $k$-means and $k$-median, using
the \emph{local search method}.  (In their case, the
non-identity-relaxed variant where the $k$ ``centers'' are known is
trivial and not interesting to solve.)

\section{Preliminaries}
Let $G=(V,E;w)$ be an undirected graph where $w: E \rightarrow R^+$
is an assignment of weights to the edges of $G$. We use $V(G)$ and $E(G)$ to denote the vertices and edges of the graph $G$ respectively.
For each edge $e \in E$, we use $w(e)$ to denote the weight of $e$. 
Similarly, for a set of edges $E' \subseteq E$,
we use $w(E')$ to denote the total weight of the edges in $E'$, 
i.e., $w(E') = \sum_{e \in E'} w(e)$.
A graph $G$ is called \emph{normalized} if $w(E)=1$.
We denote the number of (connected) components in $G$ by $comp(G)$. 
Moreover, we use $\beta(G)=|E|/|V|$ to denote the ratio of the number
of edges in $G$ to its number of vertices.

A \textit{$k$-way cut} is a partition of $V$ into $k$ disjoint, nonempty sets
$V_1, V_2, \dots, V_k$, called \emph{parts}. We use $(V_1, V_2,\dots, V_k)$ to denote the cut. 
The weight of a $k$-way cut is the total weight of the edges whose endpoints
are in different parts. We denote the weight of the cut by $w(V_1, V_2,\dots, V_k)$.

For any subset $S \subseteq E$ of edges, we use $G- S$ to denote the graph
derived from $G$ by removing the edges in $S$.
We say that a edge set $S$ is a \textit{$k$-way split} in $G$
if $comp(G-S)= (k-1)+ comp(G)$.
Therefore, $k$-way splits and $k$-way cuts are equivalent in connected graphs.
We define the \emph{separation degree} of $S$ to be $k$.  We use $\dens{G}(S)$ to denote the \emph{density} of $S$ and define it as
$$\dens{G}(S) = w(S)/(k-1) \;.$$

A graph $G$ is called \emph{$H$-minor-free} if and only if the graph $H$ does not
appear as a minor of $G$; i.e., $H$ cannot be obtained via removing and
contracting edges and removing vertices in $G$. Note that \emph{planar graphs} are a special case of minor-free graphs as they do not have $K_5$ and $K_{3,3}$ minors.  In this paper, w.o.l.g., we assume that $H$ is a complete graph.
The following lemma directly uses a result by \citeboth{thomason1984extremal} to show that the number of edges in a minor-free graph is almost linear in its number of vertices.  The proof is deferred to the appendix.

\begin{restatable}{lemma}{maxedge}
\label{lem:maxedge}
For any $H$-minor-free graph $G$, we have
$\beta(G) \le (\gamma+o(1)) |V(H)| \sqrt{\ln|V(H)|}$, 
where $\gamma= 0.319...$ is an explicit constant. 
\end{restatable} 

In the paper, we often find minimum-weight splits. The following lemma shows that for any fixed $k$, a minimum $k$-way split can be found in a polynomial-time. The proof of this lemma can be found in the appendix.

\begin{restatable}{lemma}{ksplit}
Given a graph $G$ and a parameter $k$, there exists a polynomial time algorithm that finds a minimum $k$-way split in $G$. 
\end{restatable}

\section{Beating Approximation Factor of $2$ in Minor-free Graphs}\label{sec:minor-free}
\begin{algorithm} [t!]
 \KwData{An $H$-minor-free connected graph $G$, and integer $k$}
 \begin{algorithmic} [1]
 \STATE  $C = \emptyset$. 
 \WHILE {separation degree of $C$ is at most $k-4$}
 	
 	\STATE Let $G'= G - C$ be the graph obtained
        by removing all the previous cuts from $G$. \COMMENT{Note that the separation degree of $C$ is equal to the number of connected components in $G'$.}
 	\STATE Let $C'$ be a split in $G'$ whose density is minimum among all splits with the separation degree of at most $3$.
 	\STATE $C= C \cup C'$.
 \ENDWHILE
  \STATE Let $G'= G - C$ be the graph obtained
        by removing all the previous cuts from $G$.
 \STATE Let $d$ be the separation degree of $C$. 
 \STATE Let $C'$ be a minimum $(k-d+1)$-way split in $G'$. 
 \RETURN $C \cup C'$.
 \end{algorithmic}
\caption{$2-\epsilon$ Approximation Algorithm for Minor-free Graphs}
\label{alg:greedy}
\end{algorithm}
In this section, we provide a $2-\epsilon$ approximation algorithm
for the $k$-way cut problem in minor-free graphs. Recall that in the $k$-cut problem, we are given a  connected graph, and we want to find a minimum-weight set of edges whose removal splits the graph into $k$ connected components. 

Our algorithm repeatedly finds a split in our graph and removes its edges to increase the number of connected components. The algorithm consists of two phases. In the first phase, while the number of connected components in the graph is at most $k-4$, we find minimum $2$-way and $3$-way splits, pick the one who has the lowest density and remove its edges. Every time that we remove the edges of a either $2$-way split or $3$-way split, the number of connected components increases by at most $2$. 

In the second phase of the algorithm, if the current graph has $d$ connected components, we find a minimum $(k-d+1)$-way split and remove its edges. Removing edges of this split increases the number of connected components by $k-d$, therefore our final graph has $k$ connected components. 
We show that in minor-free graphs, 
the approximation ratio of this algorithm is better than $2$ by a constant factor.

Note that for a $k \le 4$, since the number of connected components in the original graph is $1$ which is larger than $k-4$, the algorithm skips the first phase, and finds a minimum $k$-way split at its only step. When $k >4$, the algorithm repeatedly finds a minimum $2$-way split or a minimum $3$-way split, and removes its edges. Removing the edges of a $2$-way split and a $3$-way split increases the number of connected components by $1$ and $2$ respectively. Therefore, in this case the number of connected components at the end of the first phase is either $k-3$ or $k-2$. It follows that in this case, the second phase of our algorithm either finds a minimum $3$-way split or a minimum $4$-way split. 

First, we show that the density of splits picked by the algorithm in its first phase is less than $2/k$ fraction of the optimal solution. Particularly, we show that the density of minimum-density split with the separation degree of at most $3$ 
is at most $(2-\epsilon)OPT/k$ where $OPT$ is the weight of the minimum $k$-way cut, and $\epsilon$ is a positive constant depending on $|V(H)|$. Later we use this theorem to show that the approximation ratio of
our algorithm is better than $2$. Our main tool is the following lemma which shows that if in a minor-free graph density of every $2$-way split is at least $(1+\delta)/n$ for some $\delta>0$, then the weight of the maximum weighted matching is at least a constant fraction of the total weight of the graph.
\begin{lemma}
\label{lem:matching}
Given a $\delta >0$, let $G=(V,E;w)$ be a connected normalized graph with $n$ vertices such that the density of every $2$-way split is least $(1+\delta)/n$, then the weight of a maximum weighted matching in $G$ is at least $\dfrac{\delta^2}{16 \beta(G) (1+\delta)}$.
\end{lemma}
\begin{proof}
Let $A$ be the set of vertices in $G$ whose degree is at least $d$ for some integer $d$. It follows that $|A| \le 2|E|/d$. Note that $|E|=n \beta(G)$. Therefore, $|A| \le 2n \beta(G)/d$. Let $B=V \setminus A$, then $|B| \ge n(1-2\beta(G)/d)$.  Let $E_B$ be the set of edges in $E$ whose both ends are in $B$, and $E_A = E \setminus E_B$ be all other edges. Setting $d=4\beta(G)(1+\delta)/\delta$, we claim that $w(E_B) \ge \delta/2$.

For the sake of contradiction suppose that $w(E_B) < \delta/2$. Then, we have
$$ w(E_A) = w(E)- w(E_B) > w(E)-\delta/2 \,.$$
Since $G$ is normalized, we have $w(E)=1$. Therefore,
$$
w(E_A) > 1- \delta/2 \,.
$$
For every vertex $u \in B$, let $C_u$ be a split that separates $u$ from all other vertices. Then the separation degree of $C_u$ is at least $2$. Considering all $C_u$ splits, each edge in $E_B$ appears in $2$ of these splits, and each edge in $E_A$ appears in at most one of them. Thus,
$$\sum_{u \in B} w(C_u) \le 2 w(E_B) +  w(E_A) \,.$$
Note that $w(E_B)+w(E_A) = w(E) = 1$. Therefore, we have
\begin{equation}
\label{matching:less}
\begin{split}
\sum_{u \in B} w(C_u) &\le 2 w(E_B) +  w(E_A) = 1+w(E_B) < 1+ \delta/2 \,.
\end{split}
\end{equation}
On the other hand, the separation degree of every $C_u$ is at least $2$. We argue that weight of all of them is at least $(1+\delta)/n$.  If the weight of one of them is less than $(1+\delta)/n$, then the weight of a minimum $2$-way split is also less than $(1+\delta)/n$ as well as its density, which is a contraction. Therefore, we have
$$
\sum_{u \in B} w(C_u) > \dfrac{1+\delta}{n} |B| \ge (1+\delta) (1-2\beta(G)/d) \,.
$$
Substituting $d$ for $4\beta(G)(1+\delta)/\delta$, gives us
\begin{equation}
\label{matching:more}
\sum_{u \in B} w(C_u) \ge 1+\delta/2 \,. 
\end{equation}
Inequality (\ref{matching:less}) together with (\ref{matching:more}) is a contradiction. Therefore, $w(E_B) \ge \delta/2$.

Now we find a weighted matching using the following greedy algorithm. 
\begin{enumerate} [label*=\arabic*.]
\item Let $T=E_B$ be the set all the edges in $E_B$, and $\mathcal{M}= \emptyset$ be our current matching.
\item Let $e \in T$ be a edge that has the maximum weight among all the edges in $T$.
\item Add $e$ to the matching, i.e., $\mathcal{M} = \mathcal{M} \cup \{e\}$. Also, remove $e$ and all the edges which are incident to $e$ from $T$.
\item While $|T| >0$, repeat steps 2-3.
\end{enumerate}
In each step, we pick an edge that has the maximum weight in $T$, add it to our current matching, and remove all the edges which are incident to this edge from $T$. Since the degree of every vertex in $B$ is at most $d$, every time we add an edge to our matching, we remove at most $2d-1$ edges from $T$. Let $e$ be the edge picked by the algorithm in one of its steps. $e$ has the maximum weight in $T$, thus the weight of each of the removed edges in this step is at most $w(e)$. Therefore,
$$w(\mathcal{M}) \ge \dfrac{w(E_B)}{2d-1} \ge  \dfrac{\delta/2}{2d}\,.$$
Replacing $d$, we have
$$w(\mathcal{M}) \ge \dfrac{\delta/2}{2d}  = \dfrac{\delta^2}{16\beta(G)(1+\delta)}\,.$$
Therefore, we have found a matching with the total weight at least $\dfrac{\delta^2}{16\beta(G)(1+\delta)}$, and it completes the proof. Note that in minor-free graphs by Lemma \ref{lem:maxedge}, $\beta(G)$ is at most a constant, therefore we have found a matching with a constant weight in $G$.
\end{proof}

Now we are ready to prove that the density of a split whose density is minimum among all the splits with the separation degree of at most $3$,
is at most $(2-\epsilon)/k$ fraction of the weight of minimum $k$-way split.

\begin{theorem}
Given an $H$-minor-free graph $G$ and an integer $k \ge 3$, 
let $S$ be a split with the minimum density among all the splits with the separation degree of at most $3$,
Then for any $k$-way split $S_k$, we have
\begin{align*}
\dens{G}(S) \le \dfrac{(2-\epsilon) w(S_k)}{k} \,,
\end{align*}
where $\epsilon>0$ is a constant depending on $|V(H)|$. 
\label{thm:singlecut}
\end{theorem}

\begin{proof}
First consider the case that $G$ is connected.
Let $P_1, P_2,\dots, P_k$ be the components in $G-S_k$. 
For each $P_i$ let $E_i$ be set of edges whose both ends are in $P_i$. 
We contract all the edges in $E_1, E_2, \dots, E_k$ to obtain the new graph $G'=(V',E';w')$. Also, we replace parallel edges with a single edge with the weight equal to sum of them. 
The graph $G'$ has exactly $k$ vertices each corresponding to a component in $G-S_k$. Furthermore, $G'$ is $H$-minor-free since it is derived by edge contradictions from $G$.
Moreover, every split in $G'$ corresponds to a split with the same separation degree and same weight
in $G$. 
Let $v_1, v_2, \dots, v_k$ be the vertices of $G'$, 
where $v_i$ is the vertex corresponding to $P_i$. 
For each vertex $v$ in $G'$, we use $c_v$ to denote 
the weight of the edges incident to $v$. 
It follows that for every vertex $v_i$, $c_{v_i} = w(P_i, V \setminus P_i)$. 
Also, 
$$ 
w(S_k) = \dfrac{\sum_{v \in V(G')} c_v}{2} \,. 
$$
Without loss of generality, we assume that $G'$ is normalized, i.e., $w(S_k) =1$.


If there exists a $2$-way split in $G'$ with the density of at most $(2-\epsilon)/k$, then the theorem clearly holds. Otherwise, we assume that the density of every $2$-way split is greater than $(2-\epsilon)/k$. For every vertex $v$ in $G'$, the separation degree of $S_v = (\{v\}, V(G') \setminus \{v\})$ is at least $2$, and it has a weight of $c_v$. For every vertex $v$, $c_v$ is at least $(2-\epsilon)/k$, otherwise, the weight of a minimum $2$-way split is less than $(2-\epsilon)/k$, and the graph has a $2$-way split with the density less than $(2-\epsilon)/k$ which is a contradiction. 

Graph $G'$ is a normalized, and the density of every $2$-way split is at least $(2-\epsilon)/k$. Therefore, by setting $\delta=1-\epsilon$, Lemma \ref{lem:matching} implies that $G'$ has a matching with the weight at least $\alpha=\dfrac{\delta^2}{16\beta(G')(1+\delta)}$.

Let $\mathcal{M}$ be the maximum weighted matching in $G'$. We have $w'(\mathcal{M}) \ge \alpha$. Since $G'$ is connected, we have 
$$\beta(G') \ge 1- 1/k \ge 2/3 \,.$$
Setting $\epsilon = 1/(35 \beta(G'))$, it is easy to verify that $\alpha  \ge \epsilon$ while $\beta(G') \ge 2/3$. Thus, the weight of $\mathcal{M}$ is at least $\epsilon$.

For every edge $(a,b)$ in $\mathcal{M}$, let $S_{(a,b)}=(\{a\}, \{b\}, V(G') \setminus \{a, b\})$ be a split that separates $a$ and $b$ from all other vertices and each other. We claim that the weight of at least one of these splits is at most $2(2-\epsilon)/k$. For the sake of contradiction, suppose that the weight of all of them is greater than $2(2-\epsilon)/k$. Let $U$ be the set of vertices which are not in $\mathcal{M}$. Recall that for every $v \in U$, $S_v$ is a split that separates $v$ from all other vertices and its weight is at least $(2-\epsilon)/k$. Therefore,
\begin{align*}
\sum_{(a,b) \in \mathcal{M}} w(S_{(a,b)}) &+ \sum_{v \in U} w(S_v) > \dfrac{2(2-\epsilon)}{k} |\mathcal{M}| + \dfrac{(2-\epsilon)}{k} |U| \,.
\end{align*}
We have $|U| = |V'| - 2|\mathcal{M}| = k - 2|\mathcal{M}|$. Therefore,
\begin{align}
\label{ieq:3split}
\sum_{(a,b) \in \mathcal{M}} w(S_{(a,b)}) + \sum_{v \in U} w(S_v) > \dfrac{2-\epsilon}{k}(2|\mathcal{M}|+|U|) = 2-\epsilon \,.
\end{align}
Every edge which is in the matching appears in one of these splits, and every other edge appears in two of them. Recall that the weight of the matching is at least $\epsilon$. Therefore,
\begin{align*}
\sum_{(a,b) \in \mathcal{M}} w(S_{(a,b)}) + \sum_{v \in U} w(S_v) \le 2w'(E') - w'(\mathcal{M}) \le 2-\epsilon \,,
\end{align*}
which contradicts (\ref{ieq:3split}). Therefore, there exists an edge $(a,b)$ in $\mathcal{M}$ such that the weight of $S_{(a,b)}$ is at most $2(2-\epsilon)/k$. The separation degree of $S_{(a,b)}$ is at least $3$. Therefore, the weight of a minimum $3$-way split is at most $2(2-\epsilon)/k$, and its density is at most $(2-\epsilon)/k$. This completes the proof for the case $G$ is connected with $\epsilon = 1/(35 \beta(G'))$ which is a constant by Lemma \ref{lem:maxedge}.  

In case $G$ is disconnected, we construct a graph $G'$ from $G$ as follows:
\begin{itemize}
\item Add all the edges in $G$ to $G'$.
\item Create a new vertex $u$.
\item For each component in $G$, add an edge in $G'$ with the weight of 
$\infty$ 
from $u$ 
to an arbitrary vertex in this component.
\end{itemize}

This procedure produces a connected graph $G'$.
Moreover, every $k$-way split in $G$ is also a $k$-way split in $G'$,
and minimum splits in $G'$ are also minimum splits in $G$ since weight of the new edges are $\infty$, and they are not in any minimum split.
Clearly all the newly added edges will be in the same component of $G' - S_k$, and the graph obtained by contracting all the edges whose both ends are in the same component
will remain $H$-minor-free. Let $G''$ be this graph.
Similarly, the theorem holds for $\epsilon = {1}/(35 \beta(G''))$.
\end{proof}

Now that we know there always exists a $3$-way or a $2$-way split of ``acceptable'' density,
we show that the density of a split that algorithm picks in its second phase is also ``acceptable''. Recall that for a $k>4$, the separation degree of the split that Algorithm \ref{alg:greedy} picks in its second phase is either $3$ or $4$. The following claim shows how the density of minimum splits changes if we increase their separation degree. 

\begin{lemma}
\label{lemma:largercut}
Given a connected normalized graph $G=(V,E;w)$ with $k$ vertices, $\delta \ge 0$ and $h<k$, let $S$ be a $h$-way split such that $\dens{G}(S) \le (1+\delta)/k$. Then, the density of a minimum $(h+1)$-way split is at most 
$$
\dfrac{1+\delta}{k}+\dfrac{1-\delta}{h k} \,.
$$  
\end{lemma}
\begin{proof}
Let $G'=(V, E';w)$ be the graph obtained by removing all the edges in $S$ from $G$. $G'$ has $h$ connected components. Also, $w(E')= 1-w(S)$. In the following claim we show that there is a $2$-way split in $G'$ with the weight at most $2w(E')/(k-h+1)$. The proof of this claim can be found in the appendix.
\begin{restatable}{claim}{twosplit}
\label{claim:two-split}
Let $G=(V,E;w)$ be a graph with $k$ vertices and $h$ connected components where $h < k$, then there exists a $2$-way split with the weight of at most $2w(E)/(k-h+1)$.
\end{restatable}

Let $S'$ be a minimum $2$-way split in $G'$. By Claim \ref{claim:two-split}, we have
$$
w(S') \le \dfrac{2w(E')}{k-h+1} \,.
$$
Let $S''= S \cup S'$ be a $(h+1)$-way split, then we have
\begin{align*}
w(S'') &= w(S) + w(S') \le w(S) ( 1 - \dfrac{2}{k-h+1}) + \dfrac{2}{k-h+1} \,.
\end{align*}
The weight of $S$ is $w(S) = (h-1)\dens{G}(S) \le (h-1)(1+\delta)/k$. Recall that $h<k$. Therefore, $1-2/(k-h+1) \ge 0$, and we have
$$
w(S'') \le \dfrac{(h-1)(1+\delta)}{k} (1 - \dfrac{2}{k-h+1})+ \dfrac{2}{k-h+1} \,.
$$
So,
\begin{align*}
w(&S'') \le \dfrac{(h-1)(1+\delta)}{k} + \dfrac{2}{k-h+1} (1-\dfrac{(h-1)(1+\delta)}{k})  \,.
\end{align*}
Thus, the density of $S''$ is at most
\begin{align*}
\dens{G}(S'') = \dfrac{w(S'')}{h}\le \dfrac{(h-1)(1+\delta)}{h k} + \dfrac{2}{h(k-h+1)} (1-\dfrac{(h-1)(1+\delta)}{k}) \,.
\end{align*}
Therefore,
\begin{align*}
&\dens{G}(S'') \le \dfrac{(h-1)(1+\delta)}{h k} + \dfrac{2}{h(k-h+1)} (\dfrac{k-(h-1)(1+\delta)}{k}) \,.
\end{align*}
Since $k-h+1 \ge k-(h-1)(1+\delta)$, we have
\begin{align*}
\dens{G}(S'') \le \dfrac{(h-1)(1+\delta)}{h k} + \dfrac{2}{h k}= \dfrac{1+\delta}{ k} + \dfrac{1-\delta}{h k} \,.
\end{align*}
\end{proof} 
Now we can prove that the density of minimum $3$-way and minimum $4$-way splits are also less than $2/k$ fraction of the optimal solution.
\begin{claim}
\label{claim:lastsplit}
Given an $H$-minor-free graph $G$ and any $k$-way split $S_k$, the density of minimum $3$ and $4$-way splits in $G$ are at most $(2-\epsilon/2)w(S_k)/k$ and $(2-\epsilon/3)w(S_k)/k$ respectively, if the separation degree of $S_k$ is at least $3$ and $4$ respectively.
\end{claim}
\begin{proof}
As we discussed in the proof of Theorem \ref{thm:singlecut}, we can assume w.l.o.g. that $G$ is connected. We contract all the edges which are not in $S_k$ to get a new minor-free graph $G'=(V',E',w')$. The total weight of the edges in $G'$ is equal to the weight of $S_k$. W.l.o.g., we can assume that the graph $G'$ is normalized, i.e., $w(S_k) = w'(E')=1$.

First, we prove our claim for a minimum $3$-way split. By Theorem \ref{thm:singlecut}, there exists a split with a separation degree of at most $3$ and density of at most $(2-\epsilon)/k$ in $G'$. Let $S$ be this split. If the separation degree of $S$ is $3$, then our claim is proved. Otherwise, we assume that the separation degree of $S$ is $2$. Setting $\delta=1-\epsilon$, by Lemma \ref{lemma:largercut} the density of a minimum $3$-way split is at most
$$
\dfrac{1+\delta}{k} + \dfrac{1-\delta}{2k} = \dfrac{2-\epsilon}{k} + \dfrac{\epsilon}{2k} =\dfrac{2-\epsilon/2}{k} \,.
$$

Now we consider a minimum $4$-way split. We know that there exists a $3$-way split with a density of at most $(2-\epsilon/2)/k$. Setting $\delta=1-\epsilon/2$, and applying Lemma \ref{lemma:largercut}, it gives us that the density of a minimum $4$-way split is at most
\begin{align*}
\dfrac{1+\delta}{k} + \dfrac{1-\delta}{3k} = \dfrac{2-\epsilon/2}{k} + \dfrac{\epsilon/2}{3k} =\dfrac{2-\epsilon/3}{k} \,.
\qedhere
\end{align*}
\end{proof}

Note that in the first phase of  Algorithm \ref{alg:greedy}, the algorithm considers the splits whose separation degree is at most $3$, and picks the one with the lowest density. Thus, if the algorithm picks a split in its first phase, it is guaranteed that no split with a lower separation degree has a lower density. We call these splits, \emph{sparse}. Specifically, we define sparse splits as below.

\begin{definition} [Sparse split]
In a graph $G=(V,E;w)$, 
an $h$-way split $S$ is called \emph{sparse} 
if for any $h' \le h$ and $h'$-way split $S'$, the following holds.
$$ \dens{G}(S) \le \dens{G}(S') \,.$$
\end{definition}

The following theorem shows that combining some low-density sparse splits, results in a low-density split.

\begin{theorem}
\label{thm:union}
Let $G=(V,E;w)$ be a weighted graph, $S$ be a $k$-way split in $G$, and $a_1, a_2, \ldots a_l$ be integers such that $\sum_{i=1}^{l} a_i < k$. Let $C_1,C_2, \ldots, C_l$ be $l$ splits where $C_i$ is a minimum $(a_i+1)$-way split in $G_i = G- \bigcup_{j=1}^{i-1} C_j$. Let $S_i = S \setminus  \bigcup_{j=1}^{i-1} C_j$ be a $b_i$-way split in $G_i$ for every $1 \le i \le l$. Given a $\delta \ge 0$, suppose that for every $C_i$, we have
$$\dens{G}(C_i) \le \dfrac{(1+\delta) w(S_i) }{b_i} \,.$$
Also, suppose that $C_i$ is sparse in $G_i$ for every $i<l$.
Then,
$$\dens{G} ({\bigcup_{i=1}^{l} C_i}) \le \dfrac{(1+\delta) w(S)}{k} \,.$$ 
\end{theorem}

\begin{proof}
We prove this theorem by the induction on $l$. When $l=1$, the density of $C_1$ is at most $(1+\delta)w(S_1)/k$. Since $S_1=S$, the theorem holds. For the induction step suppose that $l \ge 2$, and the theorem holds for any $l-1$ splits. By induction hypothesis, for the last $l-1$ splits we have
$$\dens{G_2} ({\bigcup_{i=2}^{l} C_i}) = \dens{G} ({\bigcup_{i=2}^{l} C_i}) \le \dfrac{(1+\delta) w(S_2)}{b_2} \,,$$
since $S_2$ is a $b_2$-way split in $G_2$. It implies that
$$w ({\bigcup_{i=2}^{l} C_i}) \le \dfrac{(1+\delta) w(S_2)}{b_2} \sum_{i=2}^{l} a_i \,.$$
Let $S'= S \cap C_1$, and $S'' = C_1 \setminus S'$. Then, $S'$ is a $(p+1)$-way split in $G_2$ for some $0\le p \le a_1$. Since $C_1$ is a minimum $(a_1+1)$-way split, $S''$ is a $(a_1-p+1)$-way split in $G$. Also, for the split $S_2$, we have $S_2 = S \setminus S'$. It follows that $S_2$ is a $b_2$-way split in $G_2$ where $b_2 \ge k-p$.  We prove the induction by considering two cases on $p$.
\begin{itemize} 
\item If $p = 0$, then $b_2 \ge k$ and the separation degree of $S_2$ is at least $k$. Therefore,
\begin{align*}
w (&{\bigcup_{i=1}^{l} C_i}) \le w(C_1) + \dfrac{(1+\delta) w(S_2)}{b_2} \sum_{i=2}^{l} a_i \\
& \le \dfrac{(1+\delta) w(S)}{k} a_1 + \dfrac{(1+\delta) w(S_2)}{k} \sum_{i=2}^{l} a_i \\
& \le \dfrac{(1+\delta) w(S)}{k} \sum_{i=1}^{l} a_i \,.
\end{align*}
It implies that
$$
\dens{G} ({\bigcup_{i=1}^{l} C_i}) = \dfrac{w({\bigcup_{i=1}^{l} C_i})}{\sum_{i=1}^{l} a_i} \le \dfrac{(1+\delta) w(S)}{k} \,.
$$
This completes the induction step for this case.
\item Otherwise, $p \ge 1$, i.e., the separation degree of $S'= S \cap C_1$ is at least $2$. By sparsity of $C_1$, we have
$$
\dens{G}(C_1) \le \dens{G}(S') \Rightarrow \dfrac{w(C_1)}{a_1} \le \dfrac{w(S')}{p} \,.
$$
Therefore,
\begin{equation}
\label{ieq:wsp}
w(S') \ge \dfrac{w(C_1) \cdot p}{a_1} \,.
\end{equation}
It follows that the weight of the union of $C_1, C_2, \ldots, C_l$ is
$$
w ({\bigcup_{i=1}^{l} C_i}) \le w(C_1) + \dfrac{(1+\delta) w(S_2)}{b_2} \sum_{i=2}^{l} a_i \,.
$$
Since $S_2 = S - S'$, we have $w(S_2) = w(S) - w(S')$. Therefore,
\begin{align*}
w (&{\bigcup_{i=1}^{l} C_i}) \le w(C_1) + \dfrac{(1+\delta) (w(S)-w(S'))}{b_2} \sum_{i=2}^{l} a_i \,.
\end{align*}
By (\ref{ieq:wsp}), we have
\begin{align*}
w (&{\bigcup_{i=1}^{l} C_i}) \le w(C_1) + \dfrac{(1+\delta) (w(S)-w(C_1) \cdot p/a_1)}{b_2} \sum_{i=2}^{l} a_i \,.
\end{align*}
Let $a = \sum_{i=1}^{l} a_i$. We claim that the weight of the split $\bigcup_{i=1}^{l} C_i$ is at most $a (1+\delta) w(S)/k$. Define the function $g$ as
\begin{align*}
 g(x) =  x + \dfrac{(1+\delta) (w(S)-x \cdot p/a_1)}{b_2} \sum_{i=2}^{l} a_i \,,
 \end{align*}
 which is equal to
 \begin{align*}
 g(x) =  x + \dfrac{(1+\delta) (w(S)-x \cdot p/a_1)(a-a_1)}{b_2}\,.
 \end{align*}
Then,
$$
w ({\bigcup_{i=1}^{l} C_i}) \le g(w(C_1)) \,.
$$
Since $g$ is linear in $x$, it is sufficient to show that our claim holds for both ends of $g$. Note that $0 \le w(C_1) \le a_1 (1+\delta) w(S) /k$.

\begin{itemize} [label={$\diamond$}]
\item For $g(0)$ we have
$$
g(0) = \dfrac{(1+\delta) w(S)(a-a_1)}{b_2}  \,.
$$
Since $b_2 \ge k-p \ge k-a_1$, we have
$$
g(0) \le \dfrac{(1+\delta) w(S)(a-a_1)}{k-a_1}  \,.
$$
It is easy to verify that $(a-a_1)/(k-a_1) \le a/k$ for every $a \le k$. Therefore,
$$
g(0) \le \dfrac{(1+\delta) w(S) a}{k} \,,
$$
which proves our claim.
\item For $g(a_1 (1+\delta) w(S) /k)$ we have
\begin{align*}
g(a_1 &(1+\delta) w(S) /k) \\&= \dfrac{a_1 (1+\delta) w(S)}{k} + \dfrac{(1+\delta) (w(S)-(1+\delta)w(S) \cdot p/k)(a-a_1)}{b_2} \\
& = (1+\delta)w(S)(\dfrac{a_1}{k} + \dfrac{(1-(1+\delta)\cdot p/k)(a-a_1)}{b_2}) \,.
\end{align*}
Note that $1+ \delta \ge 1$, therefore,
\begin{align*}
g(a_1 (1+\delta) w(S) /k) & \le (1+\delta)w(S)(\dfrac{a_1}{k} + \dfrac{(1-p/k)(a-a_1)}{b_2}) \\
& =  (1+\delta)w(S)(\dfrac{a_1}{k} + \dfrac{((k-p)/k)(a-a_1)}{b_2}) \,.
\end{align*}
Also $b_2 \ge k-p$. Therefore,
\begin{align*}
g(a_1 (1+\delta) w(S) /k) & \le (1+\delta)w(S)(\dfrac{a_1}{k} + \dfrac{((k-p)/k)(a-a_1)}{k-p}) \\
&\le (1+\delta)w(S)(\dfrac{a_1}{k} + \dfrac{a-a_1}{k}) \\&= \dfrac{(1+\delta)w(S)a}{k} \,.
\end{align*}
\end{itemize}
Thus, $w({\bigcup_{i=1}^{l} C_i}) \le a(1+\delta)w(S)/k$. It follows that $\dens{G}({\bigcup_{i=1}^{l} C_i}) \le (1+\delta)w(S)/k$.
\end{itemize}
We proved the induction step for both cases, and it completes the proof for our theorem.
\end{proof}

Finally we can establish the approximation guarantee of the greedy algorithm.
\begin{theorem}\label{thm:greedy-minor-free}
The approximation ratio of Algorithm~\ref{alg:greedy} is $2-\epsilon/3$ in minor-free graphs.
\end{theorem}
\begin{proof}
If $k \le 4$, the algorithm finds the minimum $k$-way split at its only step. Therefore, the weight of the split returned by the algorithm is the optimal solution. 

Otherwise, we suppose that $k >4$. 
Let $S_{\opt}$ be a minimum $k$-way split. The algorithm successively finds a split with the separation degree of at most $3$ that has a minimum density. The only exception is the last split that it picks which is either a minimum $3$ or a minimum $4$-way split.  

Let $C_1,C_2, \ldots, C_{l}$ be the splits picked by the algorithm, $G_i = G - \bigcup_{j=1}^{i-1} C_j$, and $S_i = S_{\opt} \setminus \bigcup_{j=1}^{i-1} C_j$ be a $b_i$-way split in $G_i$. By Theorem \ref{thm:singlecut}, $\dens{G}(C_i) \le (2-\epsilon)w(S_i)/b_i$ for every $i<l$. Also by Claim \ref{claim:lastsplit}, $\dens{G}(C_l) \le (2-\epsilon/3)w(S_l)/b_l$. Also, all the splits $C_1,C_2, \dots, C_{l-1}$ are sparse. Let $C = \bigcup_{i=1}^l C_i$ be the $k$-cut returned by the algorithm. It follows from Theorem \ref{thm:union} that
$$\dens{G}(C) \le (2-\epsilon/3)w(S_{\opt})/k \,.$$
Therefore,
\begin{align*}
w(C) \le (2-\epsilon/3)w(S_{\opt}) \,.
\qedhere
\end{align*}
\end{proof}
\begin{corollary}
The approximation ratio of Algorithm~\ref{alg:greedy} is $1.9968...$ in planar graphs.
\end{corollary}
\begin{proof}
The $\epsilon$ derived from Theorem \ref{thm:singlecut} is $1/(35 \beta(G'))$ where $G'$ is a minor of $G$. If $G$ is a planar graph, then $G'$ is also planar. Therefore, $\beta(G')\le 3$, and Theorem \ref{thm:singlecut} holds for $\epsilon= 1/(35 \cdot 3) = 1/105$. Hence, the approximation ratio of Algorithm \ref{alg:greedy} in planar graphs is $2-\epsilon/3= 2- 1/315$ which is $1.9968... \,.$
\end{proof}

Putting together Theorem~\ref{thm:greedy-minor-free} with the bounds
established for $\epsilon$ in this section yields
Theorem~\ref{thm:minor-free}: there exists a polynomial-time algorithm
for $k$-cut whose approximation factor for minor-free graphs is a
constant factor smaller than $2$. The approximation guarantee is $1.9968...$ in planar graphs.

\section{Polynomial Time Approximation Scheme}\label{sec:ptas}
\begin{algorithm} [t!]
 \KwData{An $H$-minor-free connected graph $G$, integer $k$, and $\epsilon>0$}
 \begin{algorithmic} [1]
 \STATE  $C = \emptyset$. 
 \WHILE {separation degree of $C$ is at most $k-h(\epsilon)(2+1/\epsilon)$}
 	
 	\STATE Let $G'= G - C$ be the graph obtained
        by removing all the previous cuts from $G$. \COMMENT{Note that the separation degree of $C$ is equal to the number of connected components in $G'$.}
 	\STATE Let $C'$ be a split in $G'$ whose density is minimum among all splits with the separation degree of at most $h(\epsilon)$.
 	\STATE $C= C \cup C'$.
 \ENDWHILE
  \STATE Let $G'= G - C$ be the graph obtained
        by removing all the previous cuts from $G$.
 \STATE Let $d$ be the separation degree of $C$. 
 \STATE Let $C'$ be a minimum $(k-d+1)$-way split in $G'$. 
 \RETURN $C \cup C'$.
 \end{algorithmic}
\caption{PTAS for the $k$-cut Problem in Minor-free Graphs}
\label{alg:ptas}
\end{algorithm}

In this section we generalize our method to derive a polynomial time approximation scheme (PTAS) for the $k$-cut problem in minor-free graphs. Recall that in the last section we showed that approximation ratio of a natural greedy algorithm which successively removes the lowest density split with the separation degree of at most $3$ is less than $2$. 
Our main observation for proving this bound was to show that there exists a split with the separation degree of at most $3$ such that its density is at most $(2-\epsilon)/k$ fraction of the weight of a minimum $k$-way cut.

We generalize our method, and provide a PTAS for the $k$-cut problem in minor-free graphs. In this section we show that the density of minimum weighted splits converges to $1/k$ fraction of the weight of a minimum $k$-way cut if we consider splits with larger separation degrees. For an $\epsilon>0$, we first show that there exists a constant $h(\epsilon)$ such that there exists a split with the separation degree of at most $h(\epsilon)$ and the density of at most $(1+\epsilon) \opt/k$ where $\opt$ is the weight of the optimal solution. To this purpose, we use the separation theorem which shows that in every minor-free graph with $n$ vertices, the removal of $O(\sqrt{n})$ vertices, can partition the graph into two parts such that each of them has at most $2n/3$ vertices.

\begin{theorem} [\cite{alon1990separator}, \cite{lipton1979separator}]
\label{thm:separator}
Let $G$ be an $H$-minor-free graph with $n$ vertices, then there exists a separator of size of at most $c_1 \sqrt{n}$ such that $c_1$ is a constant only depending on $|V(H)|$, and removal of this separator partitions the graphs into two parts each of which has at most $2n/3$ vertices. 
\end{theorem}


The following theorem, is our main observation to derive a PTAS for the $k$-cut problem.

\begin{theorem}
\label{thm:low-density}
Given a minor $H$, and an $\epsilon>0$, there exists a constant $h(\epsilon)$ such that for any $H$-minor-free graph $G=(V,E;w)$ and any $k\ge h(\epsilon)$ and $k$-way split $S_k$ in $G$, there exists a split with the separation degree of at most $h(\epsilon)$ and the density of at most $(1+\epsilon) w(S_k)/{k} \;.
$
\end{theorem}
\begin{proof}
As we discussed in the proof of Theorem \ref{thm:singlecut}, we can assume w.l.o.g. that $G$ is connected. We contract all the edges which are not in $S_k$ to get a new minor-free graph $G'=(V',E';w')$. The total weight of the edges in $G'$ is equal to the weight of $S_k$. W.l.o.g., we can assume that the graph $G'$ is normalized, i.e., $w(S_k) = 1$. 

The following lemma is directly derived from Theorem \ref{thm:separator} which shows that for any $\delta>0$, there exists $O(k \delta)$ vertices such that removal of them partitions $G'$ into parts with the size at most $1/\delta^2$.

\begin{lemma}
\label{lem:eps-sep}
For any $H$-minor-free graph $G$ with $n$ vertices, there exists a constant $c_2$ such that for any $\delta>0$, there are $c_2 n \delta$ vertices such that removing them partitions the graph into parts with the size at most $1/\delta^2$. 
\end{lemma}
\begin{proof}
The proof is almost alike to the proof of the similar lemma in \cite{federickson1987fast}. We recursively find and remove the separator of Theorem \ref{thm:separator} in each part until its size becomes at most $1/\delta^2$. Let $b(n)$ be the number of vertices removed in an $H$-minor-free graph with $n$ vertices. The removal of the separator in Theorem \ref{thm:separator} partitions the graph into two parts such that each of them has at least $n/3$ vertices. Let $n \alpha$ be the size of the first part, then the size of the other part is at most $n(1-\alpha)$. Therefore, we have
$$
b(n) \le c_1 \sqrt{n} + b(n \alpha) + b(n(1-\alpha)) \,,
$$ 
where $1/3 \le \alpha \le 2/3$. Also, we have
$$
b(n) =0 \,,
$$
for any $n \le 1/\delta^2$.
It can be shown by induction that
$$b(n) \le c_2 n \delta - d \sqrt{n} \,,$$
for some constants $c_2$ and $d$.
\end{proof}

Note that $G'$ is $H$-minor-free. Therefore, there exists a constant $c_2$ such that Lemma \ref{lem:eps-sep} holds for $G'$. Let $\delta=\epsilon/(c_2(1+\epsilon))$, by Lemma \ref{lem:eps-sep}, there is a separator of size at most $c_2 k \delta$ such that removing it partitions $G'$ into several parts, each with the size of at most $1/\delta^2$. 

Let $P_1, P_2, \cdots, P_l$ be these parts where $P_i$ is the set of vertices in the part $i$. Let 
$P_i=\{v_{i,1}, v_{i,2}, \cdots \}$, and $C_i$ be the split that separates each vertex in $P_i$ from every other vertex in $G'$, i.e., $C_i = (\{v_{i,1}\}, \{v_{i,2}\}, \cdots , V' \setminus P_i)$. Then, the separation degree of $C_i$ is at least $|P_i|+1$. We claim that the weight of at least one of $C_i$ is at most $|P_i|(1+\epsilon)/k$. For the sake of the contradiction, suppose that the weight of every $C_i$ is greater than $|P_i|(1+\epsilon)/k$.

Note that every edge in the splits $C_i$, is either between two vertices in a same part, or between a vertex of the separator and another vertex. Therefore, each edge appears at most once in these splits. Thus,
\begin{align}
\label{ptas-ieq1}
\sum_{i=1}^{l} w'(C_i) \le 1 \,.
\end{align}
On the other hand, the weight of every $C_i$ is greater than $|P_i|(1+\epsilon)/k$. Therefore, we have
$$
\sum_{i=1}^{l} w'(C_i) > \dfrac{1+\epsilon}{k} \sum_{i=1}^{l} |P_i| \,.
$$
Since the size of the separator is at most $c_2 k \delta$, we have
\begin{align*}
\sum_{i=1}^{l} w'(C_i) & > \dfrac{1+\epsilon}{k} \sum_{i=1}^{l} |P_i| \\
& \ge \dfrac{1+\epsilon}{k} k (1-c_2 \delta) \\
& = (1+\epsilon) (1-c_2 \delta) \,.
\end{align*} 
Substituting $\delta$ with $\epsilon/(c_2(1+\epsilon))$, gives us
\begin{align}
\label{ptas-ieq2}
\sum_{i=1}^{l} w'(C_i)  > (1+\epsilon) (1-c_2 \delta) = (1+\epsilon) (1- \dfrac{\epsilon}{1+\epsilon}) = 1 \,.
\end{align}
Inequality (\ref{ptas-ieq1}) contradicts (\ref{ptas-ieq2}). Therefore, for at least one of the $C_i$, its weight is at most $|P_i|(1+\epsilon)/k$. Let $C_i$ be a split with the weight of at most $|P_i|(1+\epsilon)/k$. The separation degree of $C_i$ is at least $|P_i|+1$. Therefore, the weight of a minimum $(|P_i|+1)$-way split is at most $|P_i|(1+\epsilon)/k$, and its density is at most $(1+\epsilon)/k$. Since $|P_i|$ is at most $1/\delta^2$, the separation degree of this split is at most $1/\delta^2+1$ which is a constant. Therefore, there exists a split with a separation degree of at most $1/\delta^2+1$ and the density of at most $(1+\epsilon)/k$ which proves the theorem.
\end{proof}


Based on our observation from Theorem \ref{thm:low-density}, we modify Algorithm \ref{alg:greedy} as follows to derive a \ptas for the $k$-cut problem.
Similar to our $2-\epsilon$ algorithm, Algorithm \ref{alg:ptas} has two phases. In its first phase while the number of components in our graph is at most $k-h(\epsilon)(2+1/\epsilon)$, it picks and removes a minimum density split with a separation degree of at most $h(\epsilon)$. In its second phase, if the current graph has $d$ connected components, it finds and removes a minimum $(k-d+1)$-way split.

For a $k \le h(\epsilon)(2+1/\epsilon)$, since the number of connected components in the original graph is $1$ which is larger than $k-h(\epsilon)(2+1/\epsilon)$, the algorithm skips the first phase, and finds a minimum $k$-way split at its only step. When $k >h(\epsilon)(2+1/\epsilon)$, the algorithm repeatedly finds and removes a minimum density split with a separation degree of at most $h(\epsilon)$. Removing the edges of this split increases the number of connected components in our graph by at most $h(\epsilon)-1$ . Therefore, in this case the number of connected components at the end of the first phase of our algorithm is at most 
$$k-h(\epsilon)(2+1/\epsilon) + h(\epsilon)-1 \le k-h(\epsilon)(1+1/\epsilon)\,.
$$
It follows that the separation degree of a split picked by the second phase of the algorithm is between $h(\epsilon)(1+1/\epsilon)+1$ and $h(\epsilon)(2+1/\epsilon)$. Now, we show that the density of a minimum split with the separation degree of at least $h(\epsilon)(1+1/\epsilon)$ is also small.

\begin{lemma}
\label{lem:lastsplitptas}
Given a minor-free graph $G$ and integers $s \ge h(\epsilon)(1+1/\epsilon)$ and $k \ge s$, 
let $S$ be a minimum $s$-way split in $G$. Then,
for any $k$-way split $S_k$, we have
\begin{align*}
\dens{G}(S) \le \dfrac{(1+2\epsilon) w(S_k)}{k} \,.
\end{align*}
\end{lemma}
\begin{proof}
As we discussed in the proof of Theorem \ref{thm:singlecut}, we can assume w.l.o.g. that $G$ is connected. We contract all the edges which are not in $S_k$ to get a minor-free graph $G'=(V',E';w')$. The total weight of the edges in $G'$ is equal to the weight of $S_k$. W.l.o.g., we can assume that the graph $G'$ is normalized, i.e., $w(S_k) = w'(E')=1$.

Now we want to remove some edges in $G'$ to increase its number of connected components by $s-1$. While the number of components in $G'$ is at most $s-h(\epsilon)$, we find a split with the separation degree of at most $h(\epsilon)$ that has the minimum density. Let $C_1, C_2, \ldots, C_l$ be these splits and $(a_1+1), (a_2+1), \ldots, (a_l+1)$ be their separation degree respectively. Let $G'_i = G' - \bigcup_{j=1}^{i-1} C_j$ for every $C_i$.  It follows that $G'_i$ has $\sum_{j=1}^{i-1}a_j$ connected components and its edges is a $(k-\sum_{j=1}^{i-1}a_j)$-way split in $G'_i$. According to Theorem \ref{thm:low-density} we have
$$
\dens{G'} (C_i) \le \dfrac{(1+\epsilon)w'(G'_i)}{ k- \sum_{j=1}^{i-1}a_j} \,.
$$

Let $C= \bigcup_{i=1}^{l} C_i$, be a $s'$-way split for $G'$ where $s'=1+\sum_{i=1}^{l} a_i$. Since the number of connected components in $G'-C$ is greater than $s-h(\epsilon)$, we have $s' > s-h(\epsilon)$. Therefore, $s' > s- h(\epsilon) \ge h(\epsilon)/\epsilon$. Also, $s-s' < h(\epsilon)$. By applying Theorem \ref{thm:union} to the splits $C_1, C_2, \ldots, C_l$, we have
$$
\dens{G'}(C) \le \dfrac{(1+\epsilon)w'(G')}{k} = \dfrac{1+\epsilon}{k} \,.
$$
If $s=s'$, we have found a split with the density of at most $(1+\epsilon)/k$ and proved the theorem. Otherwise, we can assume that $s>s'$. 
By setting $\delta=\epsilon$, Lemma \ref{lemma:largercut} implies that the density of a minimum $(s'+1)$-way split is at most
$$
\dfrac{1+\epsilon}{k}+\dfrac{1-\epsilon}{s' k} < \dfrac{1+\epsilon}{k}+\dfrac{1}{s' k} \,.
$$
Since $s'$ is larger than $h(\epsilon)/\epsilon$, the density of a minimum $(s'+1)$-way split is at most
$$
\dfrac{1+\epsilon}{k}+\dfrac{1}{s' k} < \dfrac{1+\epsilon}{k}+\dfrac{\epsilon}{h(\epsilon) k} = \dfrac{1+\epsilon(1+1/h(\epsilon))}{k} \,. 
$$
Repeatedly applying Lemma \ref{lemma:largercut} implies that for any $a>0$, the density of a minimum $(s'+a)$-way split is at most
$$
\dfrac{1+\epsilon(1+a/h(\epsilon))}{k} \,.
$$
Therefore, the density of a minimum $s$-way split is at most
\begin{align*}
\dfrac{1+\epsilon(1+(s-s')/h(\epsilon))}{k} < \dfrac{1+\epsilon(1+h(\epsilon)/h(\epsilon))}{k} = \dfrac{1+2\epsilon}{k} \,.
\qedhere
\end{align*}
\end{proof}

Now we are ready to prove that Algorithm~\ref{alg:ptas} is a PTAS for the $k$-way cut in minor-free graphs.
\begin{theorem}\label{thm:main-proof}
Given an $\epsilon > 0$, the approximation ratio of Algorithm~\ref{alg:ptas} is $1+ 2 \epsilon$ in minor-free graphs.
\end{theorem}
\begin{proof}
The analysis is very similar to the that of Theorem \ref{thm:greedy-minor-free}. 
If $k \le h(\epsilon)(2+1/\epsilon) $, the algorithm finds the minimum $k$-way split at its only step. Therefore, the weight of the split returned by the algorithm is the optimal solution. Otherwise, we suppose that $k >h(\epsilon)(2+1/\epsilon)$. 

Let $S_{\opt}$ be a minimum $k$-way split. The algorithm successively finds a split with the separation degree of at most $h(\epsilon)$ that has a minimum density. The only exception is the last split that it picks which its separation degree is at least $h(\epsilon)(1+1/\epsilon)$.  

Let $C_1,C_2, \ldots, C_{l}$ be the splits picked by the algorithm, $G_i = G - \bigcup_{j=1}^{i-1} C_j$, and $S_i = S_{\opt} \setminus \bigcup_{j=1}^{i-1} C_j$ be a $b_i$-way split in $G_i$. By Theorem \ref{thm:low-density}, $\dens{G}(C_i) \le (1+\epsilon)w(S_i)/b_i$ for every $i<l$. Also by Lemma \ref{lem:lastsplitptas}, $\dens{G}(C_l) \le (1+2\epsilon)w(S_l)/b_l$. Also, all the splits $C_1,C_2, \dots, C_{l-1}$ are sparse. Let $C = \bigcup_{i=1}^l C_i$ be the $k$-cut returned by the algorithm. It follows from Theorem \ref{thm:union} that
$$\dens{G}(C) \le (1+2\epsilon)w(S_{\opt})/k \,.$$
Therefore,
\begin{align*}
w(C) \le (1+2\epsilon)w(S_{\opt}) \,. 
\qedhere
\end{align*}
\end{proof}

Theorem~\ref{thm:main-proof} establishes our second main result and proves Theorem \ref{thm:mainmain}.

\bibliographystyle{apalike} 
\bibliography{cut}
\newpage

\appendix
\section{Omitted proofs}\label{sec:appx}

\maxedge*
\begin{proof} 
For a graph $G$, we use $\eta(G)$ to denote the Hadwiger
number of $G$ which is the size of the largest complete graph that is
a minor of $G$.
Consider a complete graph $H'$ with $|V(H)|$ vertices.
This graph has $H$ as its minor, thus $G$ does not have a minor $H'$.
Therefore, $\eta(G) < |V(H')| = |V(H)|$.
It is shown in \cite{thomason1984extremal} that for every graph $G$ we have
$$
\beta(G) \le (\gamma+o(1)) (\eta(G)+1) \sqrt{\ln( \eta(G)+1)} \,, $$
where $\gamma= 0.319...$ is an explicit constant. This readily gives
\begin{align*}
\beta(G) 
\le (\gamma+o(1)) |V(H)| \sqrt{\ln |V(H)|} \,.
\qedhere
\end{align*}
\end{proof}

\ksplit*
\begin{proof}
We reduce the problem of finding a minimum $k$-way split to the minimum $k$-way cut problem. Let $G$ be the graph that we want to find a minimum $k$-way split in it. If $G$ is connected, then the $k$-way split and  $k$-way cut problems are equivalent. Otherwise, we suppose that $G$ is disconnected. 

We construct a graph $G'$ from $G$ as follows:
\begin{itemize}
\item Add all the edges in $G$ to $G'$.
\item Create a new vertex $u$.
\item For each component in $G$, add an edge in $G'$ with the weight of
$\infty$ 
from $u$ 
to an arbitrary vertex in this component.
\end{itemize}

This procedure produces a connected graph $G'$.
Every $k$-way split in $G$ is also a $k$-way split in $G'$,
and minimum splits in $G'$ are also minimum splits in $G$ since the weights of the new edges are $\infty$, and they are not in any minimum split. Therefore, the problem of finding a minimum $k$-way split in $G$ is equivalent to finding a minimum $k$-way cut in $G'$. \citeboth{goldschmidt1994polynomial} showed that this problem can be solved in a polynomial-time for a fixed $k$.
\end{proof}

\twosplit*
\begin{proof}
Let $U$ be the set of isolated vertices in $G$. Since $G$ has $h$ connected components and $h < k$, we have $|U| \le h-1$. For every vertex $v \in V \setminus U$, let $S_v$ be a split that separates $v$ from all other vertices, i.e., $S_v = ( \{ v \}, V \setminus \{ v \})$. Because $v$ is not isolated, the separation degree of $S_v$ is at least $2$. Considering all $S_v$ splits, each edge appears in two of them. Thus,
$$
\sum_{v \in V \setminus U} S_v = 2 w(E) \,.
$$
Therefore there exists a split $S_u$ such that its weight is at most
$$
\dfrac{2w(E)}{|V \setminus U|} \,.
$$ 
Since $|U| \le h-1$,
\begin{align*}
w(S_u) \le \dfrac{2w(E)}{|V \setminus U|} \le  \dfrac{2w(E)}{k-h+1}\,.
\end{align*}
Therefore, the weight of a minimum $2$-way split is at most $2w(E)/(k-h+1)$.
\end{proof}

\end{document}